\documentclass[conference]{IEEEtran}

\usepackage{url,color,cite}
\usepackage{caption}
\usepackage[pdftex]{graphicx}
\usepackage[cmex10]{amsmath}
\usepackage{amssymb}
\usepackage{amsthm}
\usepackage{float}
\usepackage{epstopdf}
\usepackage{hyperref}
\usepackage{algorithm,algpseudocode}
\usepackage{algpseudocode}
\usepackage{pifont}
\usepackage{subcaption}
\usepackage{cleveref}

\newtheorem{theorem}{Theorem}
\theoremstyle{definition}
\newtheorem{ddd}{Definition}
\theoremstyle{definition}
\newtheorem{Example}{Example}

\theoremstyle{definition}
\usepackage{mathtools}
\usepackage{breqn}
\IEEEoverridecommandlockouts

\title{Decentralized Coded Caching in Wireless Networks: Trade-off between Storage and Latency}
\begin{document}

\author{\thanks{This work was supported in part by the European Union’s Horizon $2020$ research and innovation programme under the Marie Skłodowska-Curie grant agreement No $690893$.}
\IEEEauthorblockN{Antonious M. Girgis\IEEEauthorrefmark{1},
        Ozgur Ercetin\IEEEauthorrefmark{2}, Mohammed Nafie\IEEEauthorrefmark{1}\IEEEauthorrefmark{3}, and Tamer ElBatt\IEEEauthorrefmark{1}\IEEEauthorrefmark{3}
        }
    \IEEEauthorblockA{\IEEEauthorrefmark{1}Wireless Intelligent Networks Center (WINC), Nile University, Cairo, Egypt\\
    \IEEEauthorrefmark{2} Faculty of Engineering and Natural Sciences, Sabanci University, Istanbul, Turkey.\\
        \IEEEauthorrefmark{3} Electronics and Communications Dept., Faculty of Engineering, Cairo University, Giza, Egypt\\
        Email: \{a.mamdouh@nu.edu.eg, oercetin@sabanciuniv.edu, mnafie@nu.edu.eg, telbatt@ieee.org\}
    }
\vspace{-1cm}}
\maketitle
\begin{abstract}
This paper studies the decentralized coded caching for a Fog Radio Access Network (F-RAN), whereby two edge-nodes (ENs) connected to a cloud server via fronthaul links with limited capacity are serving the requests of $K_r$ users. We consider all ENs and users are equipped with caches. A decentralized content placement is proposed to independently store contents at each network node during the off-peak hours. After that, we design a coded delivery scheme in order to deliver the user demands during the peak-hours under the objective of minimizing the \textit{normalized delivery time} (NDT), which refers to the worst case delivery latency. An information-theoretic lower bound on the minimum NDT is derived for arbitrary number of ENs and users. We evaluate numerically the performance of the decentralized scheme. Additionally, we prove the approximate optimality of the decentralized scheme for a special case when the caches are only available at the ENs.   
\end{abstract}

\section{INTRODUCTION}
Caching the most popular contents close to user terminals is a promising solution to deal with the tremendous growth of the mobile traffic in the fifth generation (5G) wireless networks. A caching system comprises of two separate phases. The first phase is the placement phase that occurs during the off-peak hours when the resources are under-utilized. During this phase, the network nodes fill their caches with the popular contents. The second phase is the delivery phase that occurs during the peak hours when the network is congested, and hence, the caches at the network nodes can be exploited to partly serve the user requests. In~\cite{maddah2014fundamental}, Maddah-Ali and Niesen have proposed a novel centralized coded caching scheme for an error-free broadcast channel. The authors showed that the network congestion can be significantly reduced by jointly designing the placement and delivery to generate coded multicast transmissions in the delivery phase. However, the centralized scheme in~\cite{maddah2014fundamental} relies on the knowledge of the number of active users in the delivery phase to design the placement phase. This limits the applicability of the centralized scheme, since the placement and the delivery phases might take place at different times. Moreover, the users in the wireless networks are characterized by their mobility in which the number of users varies in the network over time. To cope with this problem, the work in~\cite{maddah2014fundamental} is extended in~\cite{maddah2015decentralized} to present the decentralized coded caching. Recently, the concept of coded caching is studied in interference networks~\cite{xu2016fundamental,hachem2016degrees,maddah2015cache,sengupta2016cloud,tandon2016cloud}. It is shown in~\cite{maddah2015cache} that caches at the transmitters can improve the sum degrees of freedom (DoF) of the interference channel by allowing cooperation between transmitters for interference mitigation. In~\cite{xu2016fundamental}, the authors studied the benefits of caching at both the transmitters and receivers in the interference channel while focusing on the case of three receivers. The extension for an arbitrary number of receivers was studied in~\cite{hachem2016degrees}. The authors in~\cite{sengupta2016cloud,tandon2016cloud} proposed a centralized coded caching scheme for a Fog Radio Access Network (F-RAN) with caches equipped at edge-nodes (ENs) to minimize a novel metric named normalized delivery time (NDT) which measures the worst case delivery latency with respect to an interference-free baseline system in the high signal-to-noise ratio (SNR) regime.  

In contrast to prior works~\cite{xu2016fundamental,hachem2016degrees,maddah2015cache,sengupta2016cloud,tandon2016cloud} that discussed the centralized content placement, we study the decentralized caching problem for an F-RAN architecture with caches equipped at both ENs and users. We first formulate the problem for arbitrary number of ENs and users. Then, we propose the decentralized content placement in which each node stores a fraction of each file randomly and independently from other nodes. Based on the decentralized content placement, we design a coded delivery scheme that exploits the network topology to minimize the NDT for the two ENs and arbitrary number of users case. We show that the performance of the decentralized scheme is within a constant factor of the information-theoretic optimum for the case in which only ENs have caches. The performance of the decentralized scheme is evaluated numerically for the general scenario.

\section{SYSTEM MODEL}

We consider a Fog Radio Access Network (F-RAN) as depicted in Fig~\ref{fig1} comprising a cloud server that has a library of $N$ files, $\mathcal{W}\triangleq \left\{W_1,\cdots,W_N\right\}$, each of size $F$ bits. A set of $K_t$ edge-nodes, $\text{EN}_1,\cdots,\text{EN}_{K_t}$, is ready to serve the requests of $K_r$ users through a $K_t\times K_r$ Gaussian wireless interference channel with channel coefficient $h_{ki}\in \mathbb{C}$ between $\text{EN}_i$ and the $k$th user. The channels $\mathbf{H}\triangleq \left[h_{ki}\right]$, $i\in\left[K_t\right]$, $k\in\left[K_r\right]$, are independent and identically distributed according to a continuous distribution. The cloud is connected to each EN via a fronthaul link of capacity $C_F$ bits per symbol, where the symbol refers to a channel use of the downlink wireless channel. Each $\text{EN}_i$, $i\in\left[K_t\right]$, is equipped with a cache memory $V_i$ of size $M_tF$ bits, while each user, $k\in\left[K_r\right]$, is equipped with a cache memory $Z_i$ of size $M_rF$ bits. We refer to $\mu_t=M_t/N$ and $\mu_r=M_r/N$ as the normalized cache size at each EN and at each user, respectively, where $\mu_t,\mu_r\in\left[0,1\right]$.\footnote{Notations: we use $\oplus$ to denote the bitwise XOR operation. $\left[0,1\right]$ refers to the real numbers between $0$ and $1$. $\left[M\right]$ denotes the set of integers $\lbrace 1,\cdots,M\rbrace$, and $\left[a:b\right]$ denotes the set of integers $\lbrace a,\cdots,b\rbrace$ for $a\leq b$.}

% with channel coefficient $h_{ki}\in \mathbb{C}$ between the $\text{EN}_i$ and user $k$. The channels $\mathbf{H}\triangleq \left[h_{ki}\right]$, $i\in\left[K_t\right]$, $k\in\left[K_r\right]$, are independent identically distributed according to a continuous distribution. In addition, the channel state information (CSI) $\mathbf{H}$ are known to the cloud, all ENs, and all users. 
The system operates in two phases, namely \textit{placement phase} and \textit{delivery phase}. The placement phase takes place when the traffic load is low, and hence, the network resources can be utilized to fill the cache memories of each network node as a function of the content library $\mathcal{W}$ without any prior knowledge of the number of active users in the next phase, the future user demands, and channel coefficients. In the delivery phase, the $k$th user requests a file $W_{d_k}$ among the $N$ files available at the cloud, where we consider $\mathbf{d}=\left[d_1,\cdots,d_{K_r}\right]\in \left[N\right]^{K_r}$ as the demand vector. The cloud replies to user demands by sending a message $\mathbf{U}_i=\left(U_i\left(t\right)\right)_{t=1}^{T_F}$ of block length $T_F$ to $\text{EN}_i$ via the fronthaul link, where each fronthaul message $\left(\mathbf{U}_i\right)_1^{K_t}$ cannot exceed $T_F C_F$ bits due to the fronthaul capacity limitations. Each $\text{EN}_i$, $i\in\left[K_t\right]$, has an encoding function mapping the cache contents $V_i$, the fronthaul message $\mathbf{U}_i$, the demand vector $\mathbf{d}$, and the channels $\mathbf{H}$ to a message $\mathbf{X}_i\triangleq\left(X_i\left(t\right)\right)_{t=1}^{T_E}$ of block length $T_E$. $\text{EN}_i$ transmits message $\mathbf{X}_i$ over the interference channel to $K_r$ users under an average transmit power constraint $\frac{1}{T_E} \|\mathbf{X}_i \|^2 \leq P$. Each user $k$ implements a decoding function to estimate the requested file $\hat{W}_{d_k}$ from the cache contents $Z_k$ and the received message $\mathbf{Y}_k\triangleq \left(Y_k\left(t\right)\right)_{t=1}^{T_E}$.
\begin{equation}
Y_k\left(t\right)=\sum_{i=1}^{K_t}h_{ki} X_i\left(t\right)+n_k\left(t\right),\qquad t\in \left[T_E\right]
\end{equation}
where $Y_k\left(t\right)\in\mathbb{C}$ is the received signal of the $k$th user at time $t$, and $n_k\left(t\right)\sim \mathcal{CN}\left(0,1\right)$ denotes the complex Gaussian noise at the $k$th user. For a given coding scheme (caching, cloud encoding, EN encoding, and user decoding functions), the transmission rate of the system is defined as $K_rF/T$, where $T$ is the end-to-end latency. We say that a coding scheme is feasible, if and only if, we have
\begin{equation*}
\max_{\mathbf{d}\in\left[N\right]^{K_r}} \max_{k\in\left[K_r\right]}\mathbb{P}\left(\hat{W}_{d_k}\neq W_{d_k}\right)\rightarrow 0\ \text{as}\ F\rightarrow \infty
\end{equation*}
which is the worst-case probability of error over all possible demands $\mathbf{d}$ and over all users. Our goal in this paper is to design a decentralized coding scheme to minimize the end-to-end latency $T$ for delivering user demands $\mathbf{d}\in\left[N\right]^{K_r}$ in the delivery phase. In the following, we define the normalized delivery time (NDT) first discussed in~\cite{tandon2016cloud} as a performance metric for any coding scheme.
\begin{ddd}
The normalized delivery time (NDT) for any feasible coding scheme with a given normalized cache size $\mu_t$, $\mu_r$, and fronthaul capacity $C_F=r\log\left(P\right)$, is defined as
\begin{equation}
\delta\left(\mu_t,\mu_r,r\right)=\lim_{P\rightarrow \infty}\lim_{F\rightarrow \infty}\sup \frac{\underset{\mathbf{d}}{\max}\ \mathbb{E}_{\mathbf{H}}\left(T\right)}{F/\log\left(P\right)}
\end{equation}
where $\mathbb{E}$ denotes the expectation with respect to the channel matrix $\mathbf{H}$, and $r$ measures the multiplexing gain of the fronthaul links. Moreover, we define the minimum NDT for a given tuple $\left(\mu_t,\mu_r,r\right)$, as
\begin{equation}~\label{NDT}
\delta^{*}\left(\mu_t,\mu_r,r\right)\triangleq \inf\left\{\delta\left(\mu_t,\mu_r,r\right):\delta\left(\mu_t,\mu_r,r\right)\ \text{is feasible}\right\}
\end{equation}
\vspace{-5mm}
\end{ddd}
The NDT $\delta\left(\mu_t,\mu_r,r\right)$ represents the worst case delay to serve any possible user demands $\mathbf{d}\in\left[N\right]^{K_r}$ normalized with respect to an interference-free baseline system of transmission rate $\log\left(P\right)$ in the high SNR regime. 

\begin{figure} [t]
\centering
\includegraphics[scale=0.2]{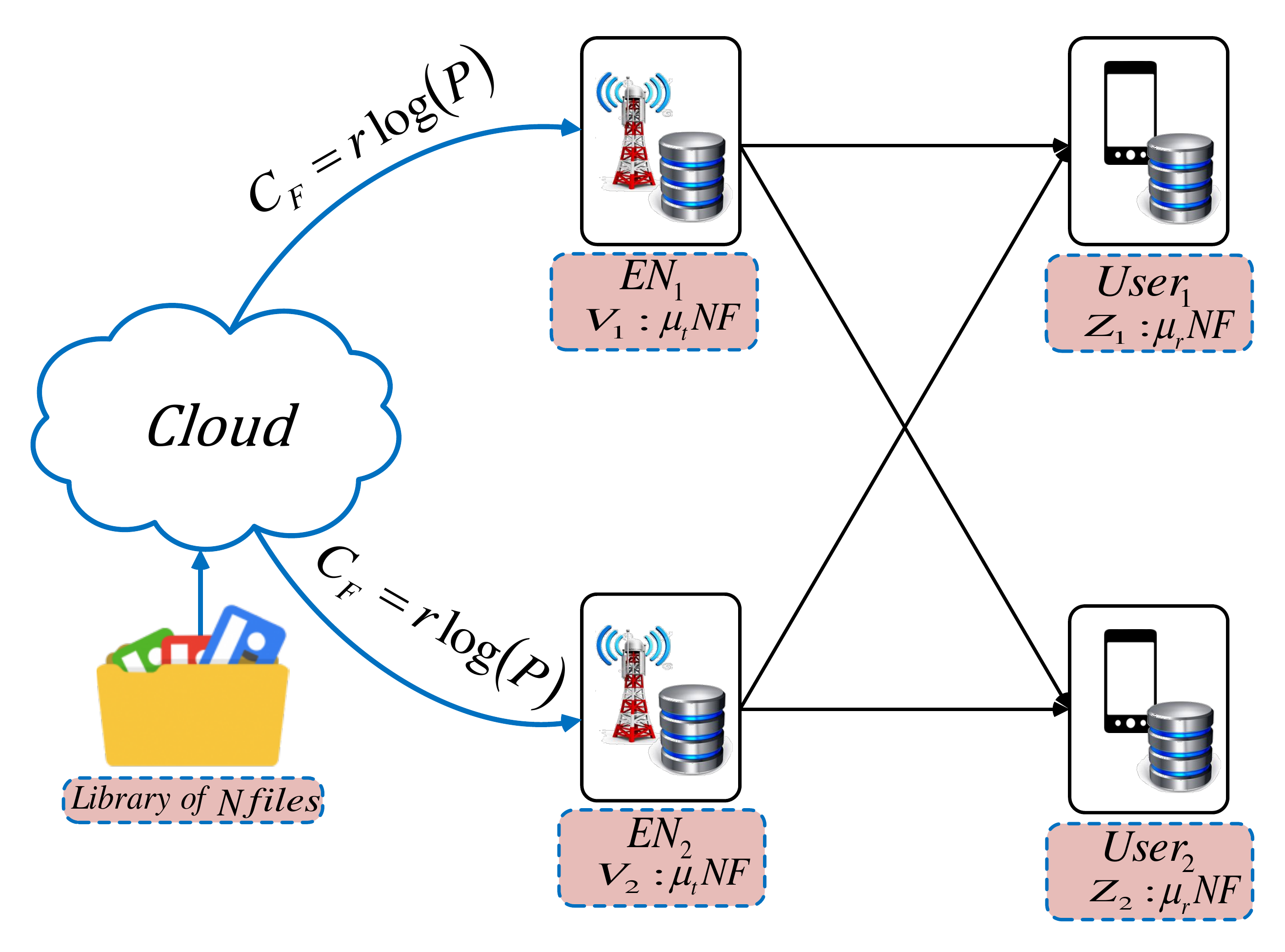}
\caption{A Fog Radio Access Network (F-RAN) with $K_t=2$ edge-nodes and $K_r=2$ users.}
\label{fig1}
\vspace{-0.5cm}
\end{figure}

We point out that the end-to-end latency for an F-RAN depends on two transmission delays, \textit{fronthaul delay} $T_F$, and \textit{edge delay} $T_E$. The fronthaul delay $T_F$ represents the transmission time to deliver the cloud messages to ENs, while the edge delay $T_E$ represents the transmission time to deliver the messages of ENs to end users. In  order to obtain the fronthaul-NDT $\delta_F$, we substitute in~\eqref{NDT} $T=T_F$. Similarly, we get the edge-NDT $\delta_E$ by substituting $T=T_E$ in~\eqref{NDT}. Throughout the paper, we consider two types of transmissions. 

1) \textbf{Serial transmission:} Fronthaul and edge transmissions occur consecutively. The cloud first transmits its messages $\left\{\mathbf{U}_i\right\}_{1}^{K_t}$ to ENs that begin delivering their messages $\mathbf{X}_i$ to the users only after receiving the fronthaul message. As a result, the overall latency is the sum of the fronthaul and edge delays. Thus, the NDT of system with serial transmission is defined as \vspace{-3mm}
\begin{equation}\label{ser}
\delta_S=\delta_F+\delta_E.\vspace{-1mm}
\end{equation}

2) \textbf{Pipelined transmission:} In~\cite{sengupta2016cloud}, the authors introduced the concept of pipelined transmission, where the fronthaul and edge transmissions take place simultaneously. In other words, ENs do not wait for the fronthaul transmission to be completed in order to start transmitting their messages to the users. Accordingly, the NDT of the system for the pipelined transmission is given by 
\begin{equation}\label{pipe}
\delta_P=\max\left\{\delta_F,\delta_E\right\}
\end{equation}
which is the maximum between the fronthaul-NDT and the edge-NDT. In the following, we characterize the trade-off between the normalized cache size $\mu_t$, $\mu_r$ and the NDT of the system for the decentralized caching scheme with both serial transmission $\delta_S$ and pipelined transmission $\delta_P$.

\section{DECENTRALIZED CODED CACHING}
\label{Decentralized}

In this section, we introduce the decentralized coded caching for the $2\times K_r$ F-RAN caching problem. First, we present the decentralized content placement. Then, we design the delivery scheme to serve any user demands in the delivery phase.

\subsection{Decentralized Content Placement} In the placement phase each user $k\in\left[K_r\right]$ independently stores $\mu_rF$ bits from each file $W_j\in\mathcal{W}$ uniformly at random. In a similar way, each $\text{EN}_i$, $i\in\left[K_t\right]$, independently stores $\mu_tF$ bits from each file $W_j\in\mathcal{W}$ uniformly at random. Therefore, the total number of bits stored at each user, and at each EN is $\mu_rNF$ bits and $\mu_tNF$ bits, respectively. Thus, the cache memory size constraint is not violated at each node. Moreover, each bit of a file has probability $\mu_r$ to be cached at a specific user and probability $\mu_t$ to be cached at a specific EN. Accordingly, each file $W_j\in\mathcal{W}$ is partitioned into fragments $W_{j,\mathcal{S}_t,\mathcal{S}_r}$ for $\mathcal{S}_t\subset\left[K_t\right]$ a subset of ENs, and $\mathcal{S}_r\subset\left[K_r\right]$ a subset of users, in which the fragment $W_{j,\mathcal{S}_t,\mathcal{S}_r}$ denotes the bits of file $W_j$ stored exclusively in the cache memory of each EN $i\in\mathcal{S}_t$ and each user $k\in\mathcal{S}r$, but these bits do not exist in the caches of other nodes. For $\mid \mathcal{S}_t\mid=t$, $\mid\mathcal{S}_r\mid=i$, each file $W_j$ is split into $\sum_{t=0}^{K_t}\sum_{i=0}^{K_r}{K_t\choose t}{K_r\choose i}= 2^{K_t+K_r}$ fragments.
\begin{Example}~\label{Ex1}
Consider an F-RAN with $K_t=2$ ENs, $K_r=2$ users, and $N=2$ files denoted by $\mathcal{W}\triangleq\left\{A,B\right\}$. According to the decentralized content placement, we can represent the file $A$ as follows\footnote{We write $A_{0,1}$ as a shorthand for $A_{\lbrace\phi\rbrace,\lbrace1\rbrace}$, and $A_{12,1}$ as a shorthand for $A_{\lbrace1,2\rbrace,\lbrace1\rbrace}$} $A=\Big(A_{0,0},A_{0,1},A_{0,2},A_{0,12},A_{1,0},A_{1,1},A_{1,2}$ $,A_{1,12},A_{2,0},A_{2,1},A_{2,2},A_{2,12},A_{12,0},A_{12,1},A_{12,2},A_{12,12}\Big)$
such that $A_{0,0}$ consists of bits of file $A$ that are not stored at any node, and $A_{2,12}$ consists of bits of file $A$ in the cache memory of all nodes except $\text{EN}_1$. The same follows for file $B$. By strong law of large numbers, the size of fragment $W_{j,\mathcal{S}_t,\mathcal{S}_r}$ is approximately equal to
\begin{equation}\label{eqn1}
\mid W_{j,\mathcal{S}_t,\mathcal{S}_r}\mid\ \approx\ \mu_t^{\mid\mathcal{S}_t\mid}\left(1-\mu_t\right)^{K_t-\mid\mathcal{S}_t\mid}\mu_r^{\mid\mathcal{S}_r\mid}\left(1-\mu_r\right)^{K_r-\mid\mathcal{S}_r\mid}F
\end{equation}
with probability approaching one for large file size $F$. In our example, $\mid A_{1,2}\mid \approx \mu_t\left(1-\mu_t\right)\mu_r\left(1-\mu_r\right)F$.
\end{Example}

\subsection{Coded Delivery}\label{Coded} Consider any demand vector $\mathbf{d}\in\left[N\right]^{K_r}$. In Example~\ref{Ex1}, let the first user request file $A$ and the second user request file $B$. Therefore, the first user needs fragments $\left(A_{0,0},A_{0,2},A_{1,0},A_{1,2},A_{2,0},A_{2,2},A_{12,0},A_{12,2}\right)$ to decode file $A$, and the second user needs fragments $\left(B_{0,0},B_{0,1},B_{1,0},B_{1,1},B_{2,0},B_{2,1},B_{12,0},B_{12,1}\right)$ to decode file $B$. In the following, we explain the delivery scheme to send these fragments. The proposed delivery scheme is divided into five stages. In each stage, the cloud and ENs exploit the wireless channel to deliver a group of fragments of the requested files depending on the cache contents at the nodes. 
 
1) \textit{Delivery of fragments $W_{d_k,\phi,\phi}$}: Since these $K_r$ fragments are not available at any nodes, the cloud sends $\lbrace W_{d_k,\phi,\phi}\rbrace_{k=1}^{K_r}$ to each EN via the fronthaul links. Hence, each EN can perform cooperative zero-forcing (ZF) beamforming on the interference channel, where the system becomes a Multiple-Input-Single-Output (MISO) broadcast channel with $K_t=2$ transmit antennas and $K_r$ receivers. In general, the sum DoF for a broadcast channel with $K_t$ transmit antennas and $K_r$ receivers is $\min\lbrace K_t,K_r\rbrace$~\cite{weingarten2006capacity}. In our case $DoF=2$. Thus, the edge transmission latency to deliver these fragments in the high SNR regime is $T_E^{\left(1\right)}=\frac{K_r\mid W_{d_k,\phi,\phi}\mid F}{2\log\left(P\right)}$. By using the \textit{soft-transfer mode} on fronthaul links as proposed in~\cite{simeone2009downlink}, the cloud implements ZF-beamforming and quantizes the encoded signal to be transmitted to each EN. Therefore, the fronthaul transmission latency equals to the edge transmission latency multiplied by the ratio between the transmission rates $\log\left(P\right)/r\log\left(P\right)$, yielding a fronthaul transmission latency $T_F^{\left(1\right)}=\frac{K_r\mid W_{d_k,\phi,\phi}\mid F}{2r\log\left(P\right)}$. Thus, fronthaul-NDT and edge-NDT for this stage are given by
\begin{equation}
\begin{aligned}
\delta_F^{\left(1\right)}=\frac{K_r}{2r}\mid W_{d_k,\phi,\phi}\mid\approx\frac{K_r}{2r}\left(1-\mu_t\right)^2\left(1-\mu_r\right)^{K_r}\\
\delta_E^{\left(1\right)}=\frac{K_r}{2}\mid W_{d_k,\phi,\phi}\mid \approx\frac{K_r}{2}\left(1-\mu_t\right)^2\left(1-\mu_r\right)^{K_r}
\end{aligned}
\end{equation} 
In our example, the first user receives fragment $A_{0,0}$ and the second user receives fragment $B_{0,0}$ at the end of this stage.

2) \textit{Delivery of fragments $W_{d_k,\phi,\mathcal{S}_r}$, $\mid \mathcal{S}_r\mid\geq 1$}: These fragments are available at least at one user and at none of ENs. For each subset $\mathcal{S}_r\subset \left[K_r\right]$ of cardinality $\mid \mathcal{S}_r\mid=i\in\left[2:K_r\right]$, the cloud sends a multicast message $\oplus_{k\in\mathcal{S}_r} W_{d_k,\phi,\mathcal{S}_r\setminus{\lbrace k\rbrace}}$ to each EN, where $W_{d_k,\phi,\mathcal{S}_r\setminus{\lbrace k\rbrace}}$ refers to the bits of file $W_{d_k}$ requested by user $k$ which are available exclusively at users $\mathcal{S}_r$ except user $k$. Then ENs deliver the fronthaul message to $K_r$ users. Since the $k$th user has the bits $W_{d_{\tilde{k}},\phi,\mathcal{S}_r\setminus{\lbrace\tilde{k}\rbrace}}$ for $\tilde{k}\in\mathcal{S}_r$, the $k$th user can recover its needed bits $W_{d_k,\phi,\mathcal{S}_r\setminus{\lbrace k\rbrace}}$ from the multicast transmission. There are $K_r\choose i$ subsets of cardinality $i$, so the rate of this transmission is given by \vspace{-0.5cm}

\small
\begin{IEEEeqnarray}{l}
R^{\left(2\right)}=\sum_{i=2}^{K_r} {K_r\choose i} \left(1-\mu_t\right)^2 \mu_r^{i-1}\left(1-\mu_r\right)^{K_r-i+1}F\nonumber\\
=\frac{\left(1-\mu_t\right)^2\left(1\mu_r\right)}{\mu_r}\Big[\sum_{i=0}^{K_r}{K_r\choose i}-\left(1-\mu_r\right)^{K_r}\\
\qquad\qquad\qquad\qquad\qquad\qquad\qquad-K_r\mu_r\left(1-\mu_r\right)^{K_r-1}\Big]\nonumber\\
=\frac{\left(1-\mu_t\right)^2\left(1-\mu_r\right)}{\mu_r}\left[1-\left(1-\mu_r\right)^{K_r}-K_r\mu_r\left(1-\mu_r\right)^{K_r-1}\right] F\nonumber
\end{IEEEeqnarray}  
\normalsize
where the expected size of $\oplus_{k\in\mathcal{S}_r} W_{d_k,\phi,\mathcal{S}_r\setminus{\lbrace k\rbrace}}$ with $\mid \mathcal{S}_r\mid=i$ is $\left(1-\mu_t\right)^2 \mu_r^{i-1}\left(1-\mu_r\right)^{K_r-i+1}F$. As a result, fronthaul and edge transmission delays are $T_F^{\left(2\right)}=\frac{R^{\left(2\right)}}{r\log\left(P\right)}$, $T_E^{\left(2\right)}=\frac{R^{\left(2\right)}}{\log\left(P\right)}$, respectively. Moreover, the fronthaul-NDT and edge-NDT are given by
\begin{equation}
\delta_F^{\left(2\right)}=\frac{R^{\left(2\right)}}{rF},  \qquad \delta_E^{\left(2\right)}=\frac{R^{\left(2\right)}}{F}
\end{equation} 
In our example, the transmission in this stage is $\left(A_{0,2}\oplus B_{0,1}\right)$, so that the first user can recover fragment $A_{0,2}$ from the received signal and cached fragment $B_{0,1}$. Similarly, the second user can recover fragment $B_{0,1}$ from the received signal and cached fragment $A_{0,2}$. 

3) \textit{Delivery of fragments $W_{d_k,\mathcal{S}_t,\mathcal{S}_r}$, $\mid \mathcal{S}_r\mid\geq 1$}: In this stage, each subset $\mathcal{S}_t\subset\left[K_t\right]$ of ENs, $\mid\mathcal{S}_t\mid\in\left\{1,2\right\}$, transmits a multicast signal $\oplus_{k\in\mathcal{S}_r} W_{d_k,\mathcal{S}_t,\mathcal{S}_r\setminus{\lbrace k\rbrace}}$ from their cache contents to a subset $\mathcal{S}_r\subset \left[K_r\right]$ of users for $\mid\mathcal{S}_r\mid\geq 2$, so that each user $k\in\mathcal{S}_r$ can recover the needed fragment $W_{d_k,\phi,\mathcal{S}_r\setminus{\lbrace k\rbrace}}$ from the received signal and its cache contents. Thus, the rate of this transmission is given by

\small
\begin{IEEEeqnarray}{l}
R^{\left(3\right)}=\sum_{t=1}^{2}\sum_{i=2}^{K_r} {2\choose t} {K_r\choose i} \mu_t^{t}\left(1-\mu_t\right)^{2-t} \mu_r^{i-1}\left(1-\mu_r\right)^{K_r-i+1}F\nonumber\\
=\frac{\left(1-\left(1-\mu_t\right)^2\right)\left(1-\mu_r\right)}{\mu_r}\Big[1-\left(1-\mu_r\right)^{K_r}\\
\qquad\qquad\qquad\qquad\qquad\qquad\qquad-K_r\mu_r\left(1-\mu_r\right)^{K_r-1}\Big]F.\nonumber
\end{IEEEeqnarray}
\normalsize
As a result, the edge transmission latency is $T_E^{\left(3\right)}=\frac{R^{\left(3\right)}}{\log\left(P\right)}$, yielding the the fronthaul-NDT and edge-NDT
\begin{equation}
\delta_F^{\left(3\right)}=0,  \qquad \delta_E^{\left(3\right)}=\frac{R^{\left(3\right)}}{F}.
\end{equation} 
In our example, $\text{EN}_1$ first transmits a message $\left(A_{1,2}\oplus B_{1,1}\right)$, then, $\text{EN}_2$ transmits a message $A_{2,2}\oplus B_{2,1}$. Finally, both of ENs cooperatively transmit a message $A_{12,2}\oplus B_{12,1}$. Hence, the first user obtains fragments $A_{1,2},A_{2,2},A_{12,2}$, and the second user obtains fragments $B_{1,1},B_{2,1},B_{12,1}$ at the end of this stage. 

\begin{figure}[t!]
 \begin{subfigure}{0.49\linewidth}
  \centerline{\includegraphics[scale=0.3]{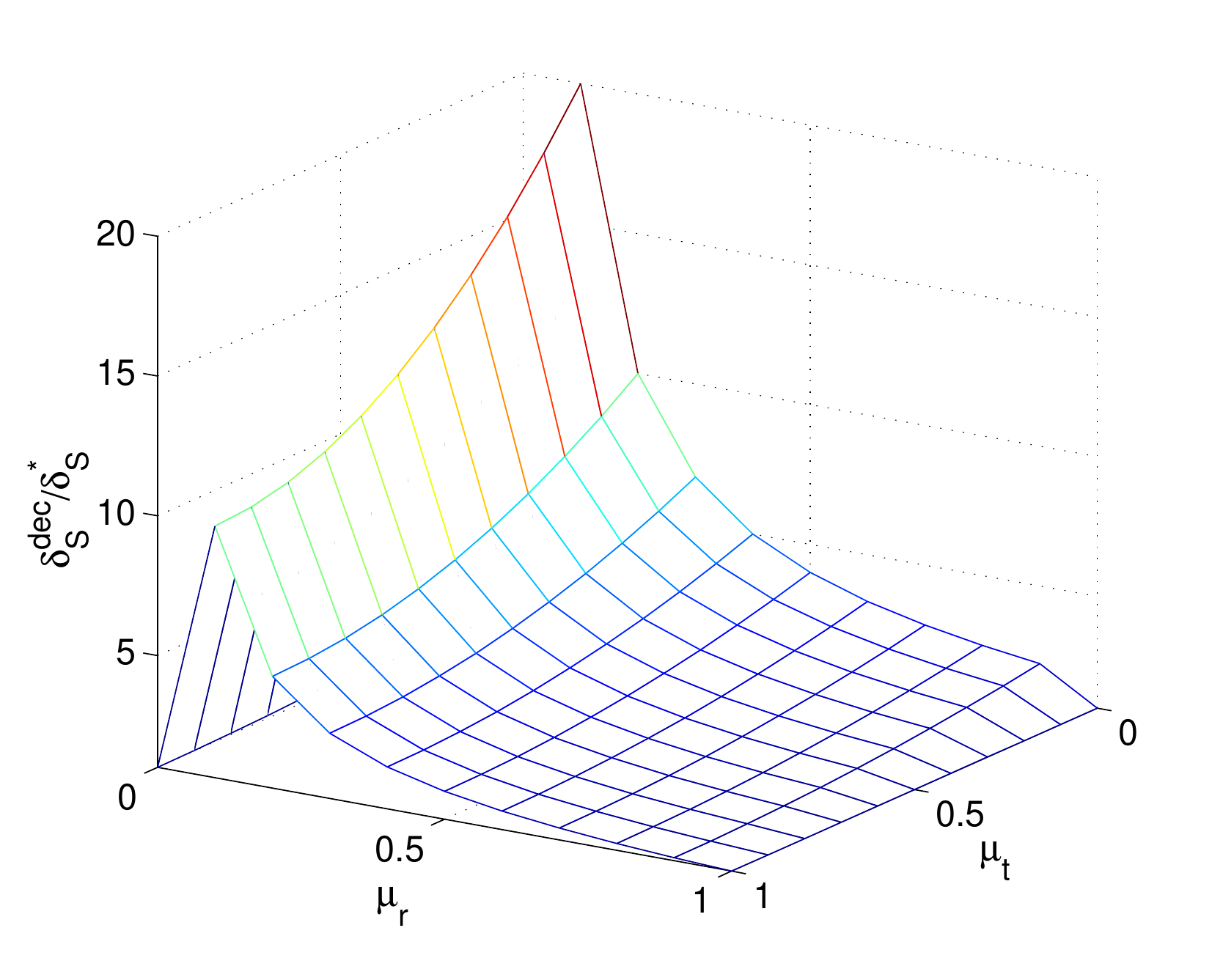}}
 \caption{Serial transmission}\label{figA}
 \end{subfigure}
 \begin{subfigure}{0.49\linewidth}
  \centerline{\includegraphics[scale=0.3]{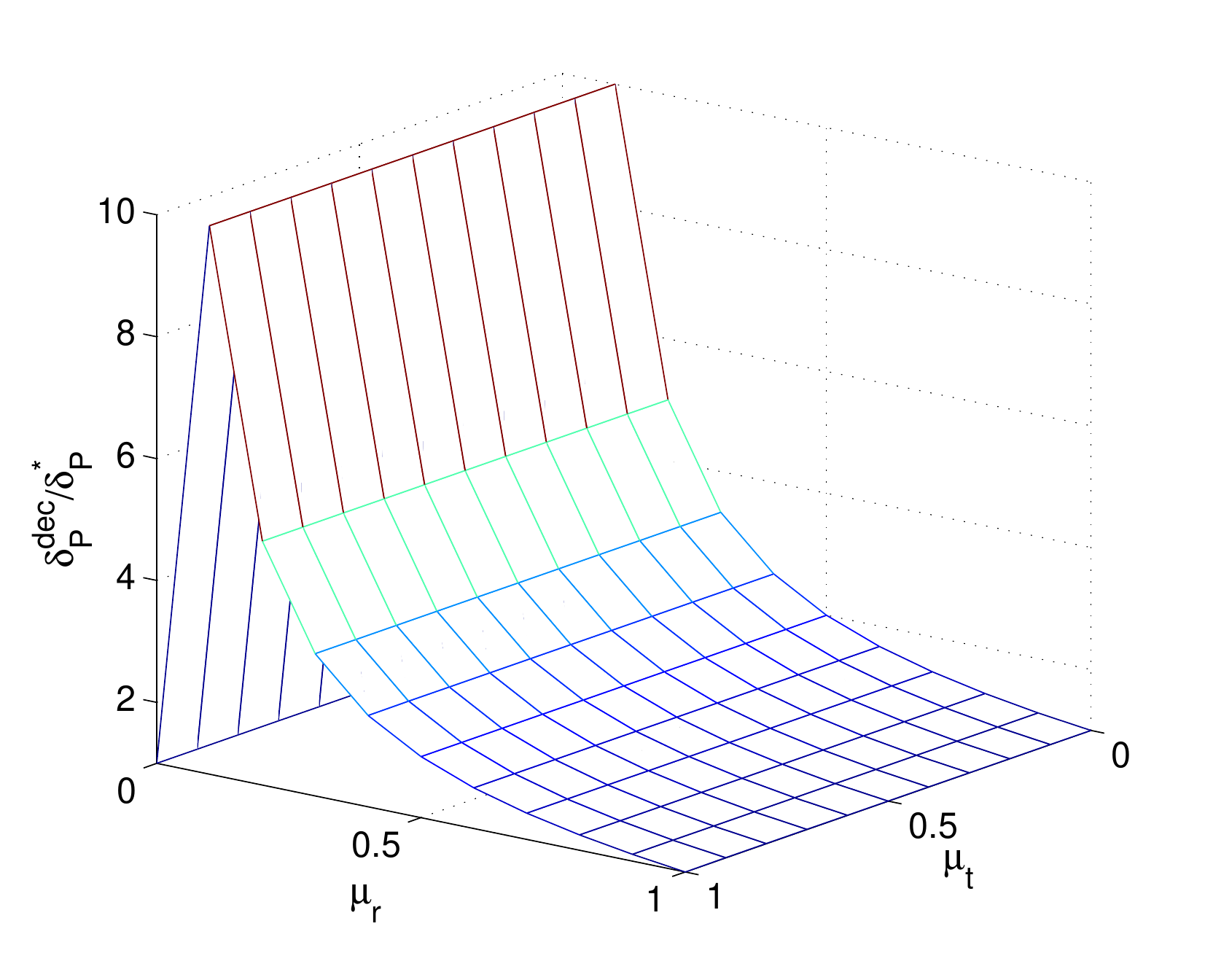}}
\caption{Pipelined transmission}\label{figB}
 \end{subfigure}
\caption{The gap between the decentralized NDT in Theorem~\ref{Th1} and the lower bound in Theorem~\ref{Th2} for $r=1$ and $K_r=100$.}
\label{Fig2}
\vspace{-0.4cm}
\end{figure}

4) \textit{Delivery of fragments $W_{d_k,\left\{1,2\right\},\phi}$}:Each of the two ENs has these $K_r$ fragments. Hence, each EN can implement ZF beamforming on the interference channel to transmit these fragments with $DoF=2$. Thus, the edge transmission latency is $T_E^{\left(4\right)}=\frac{K_r\mid W_{d_k,\lbrace1,2\rbrace,\phi}\mid F }{2\log\left(P\right)}$, yielding the fronthaul-NDT and edge-NDT 
\begin{equation}
\delta_F^{\left(4\right)}=0,  \ \delta_E^{\left(4\right)}=\frac{K_r}{2}\mid W_{d_k,\lbrace1,2\rbrace,\phi}\mid\approx \frac{K_r}{2}\mu_t^2\left(1-\mu_r\right)^{K_r}.
\end{equation}
In our example, the first user receives fragment $A_{12,0}$, and the second user receives fragment $B_{12,0}$ at the end of this stage.

5) \textit{Delivery of fragments $W_{d_k,\mathcal{S}_t,\phi}$, $\mid \mathcal{S}_t\mid=1$}: In this stage, we propose two methods for transmitting these fragments, where the ENs choose the method that minimizes the overall NDT as we will elaborate in Theorem~\ref{Th1}.\\
a) Each EN has a dedicated fragment to each user. Hence, the system can be treated as "$K_t\times K_r$ X-channel", where the sum DoF is $\frac{K_tK_r}{K_t+K_r-1}$ by applying the interference alignment (IA) scheme~\cite{cadambe2009interference}. In our case the sum DoF is $\frac{2K_r}{K_r+1}$. Thus, the edge transmission latency is $T_E^{\left(5\right)}=\frac{2K_r\mid W_{d_k,\lbrace1\rbrace,\phi}\mid F}{\frac{2K_r}{K_r+1}\log\left(P\right)}$ leading to the fronthaul-NDT and edge-NDT
\begin{equation}
\delta_F^{\left(5,a\right)}=0,\ \delta_E^{\left(5,a\right)}= \left(K_r+1\right)\mu_t\left(1-\mu_t\right)\left(1-\mu_r\right)^{K_r}.
\end{equation}
b) Another achievable transmission scheme for this stage is that the cloud transmits fragments $\left\{ W_{d_k,\lbrace1\rbrace,\phi}\right\}_{1}^{K_r}$ to $\text{EN}_2$ and fragments $\left\{ W_{d_k,\lbrace2\rbrace,\phi}\right\}_{1}^{K_r}$ to $\text{EN}_1$ via the fronthaul links. Thus, both ENs have $2K_r$ fragments and the system becomes a broadcast channel with DoF $\min\left\{K_t,K_r\right\}$ by applying ZF beamforming on the interference channel. Hence the fronthaul-NDT and edge-NDT are given by
\begin{equation}
\begin{aligned}
\delta_F^{\left(5,b\right)}=\frac{K_r}{r}\mid W_{d_k,\lbrace1\rbrace,\phi} \mid=\frac{K_r}{r}\mu_t\left(1-\mu_t\right)\left(1-\mu_r\right)^{K_r}\\
\delta_E^{\left(5,b\right)}=K_r\mid W_{d_k,\lbrace1\rbrace,\phi} \mid=K_r\mu_t\left(1-\mu_t\right)\left(1-\mu_r\right)^{K_r}.
\end{aligned}
\end{equation} 

\begin{figure} [t!]
\centering
\includegraphics[scale=0.34]{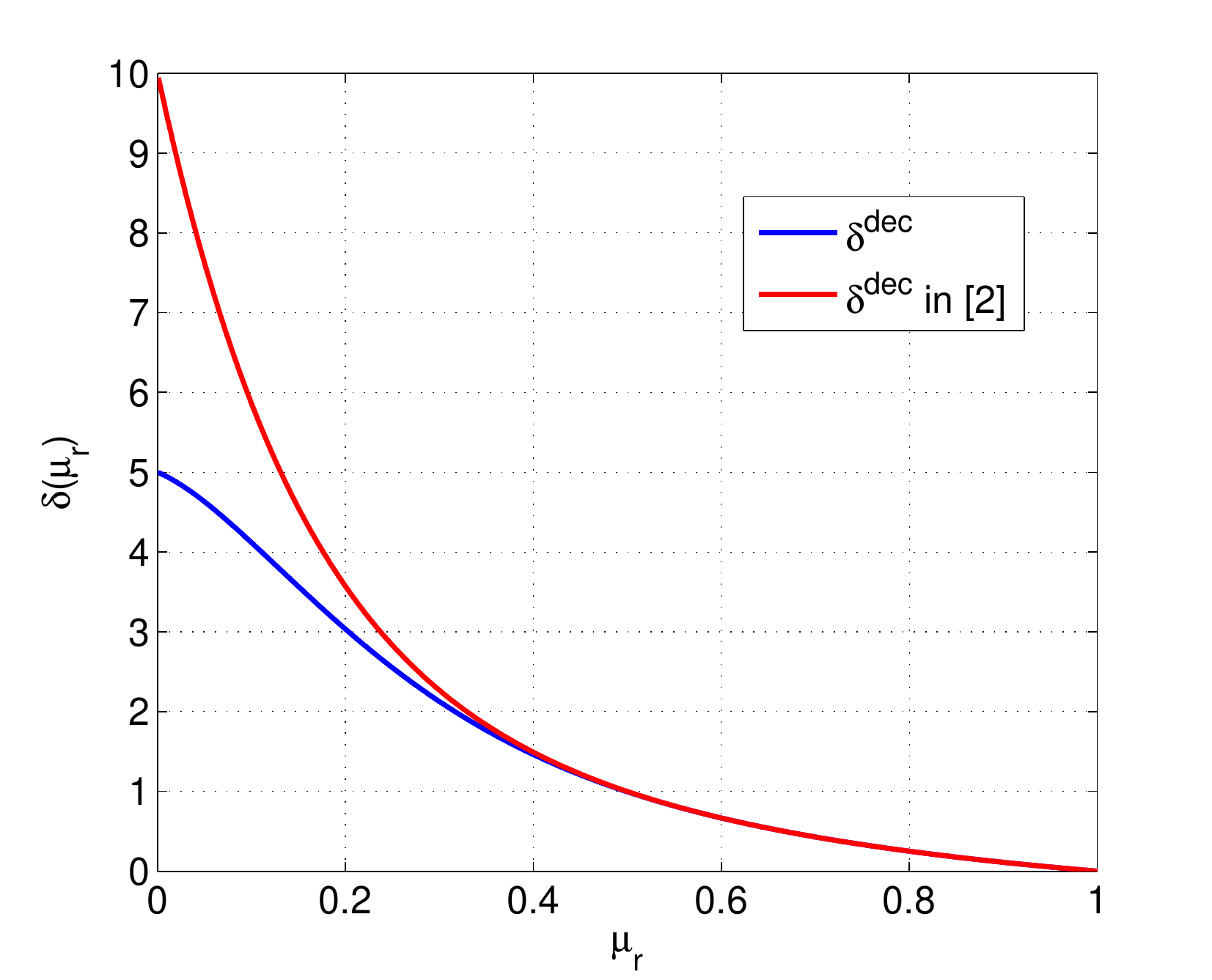}
\caption{Comparison between the decentralized scheme in~\eqref{eqn2} and the decentralized scheme~\cite{maddah2015decentralized} for $K_r=10$.}
\label{Fig3}
\vspace{-0.5cm}
\end{figure}
The following theorem gives the NDT of the proposed decentralized scheme for the serial and pipelined transmissions.
\begin{theorem}~\label{Th1}
For the $2\times K_r$ F-RAN with decentralized caching placement, each EN having a normalized cache size $\mu_t\in\left[0,1\right]$, each user having a normalized cache size $\mu_r\in\left[0,1\right]$, and the fronthaul gain $r> 0$, the proposed coded delivery scheme achieves NDT for serial transmission  
\begin{equation}\label{eqn6}
\delta_S^{dec}\triangleq \left\{\begin{array}{l l}
\delta_F^{\left(a\right)}+\delta_E^{\left(a\right)} & 0<r\leq K_r\\
\delta_F^{\left(b\right)}+\delta_E^{\left(b\right)} & r>K_r\\
\end{array}
\right.
\end{equation}
and achieves NDT for pipelined transmission 
\begin{equation}~\label{eqn7}
\delta_{P}^{dec}=\min\left\{\max\left\{\delta_F^{\left(a\right)},\delta_E^{\left(a\right)}\right\},\max\left\{\delta_F^{\left(b\right)},\delta_E^{\left(b\right)}\right\} \right\}
\end{equation}
\small
\begin{equation}\label{eqn3}
\begin{array}{l}
\delta_F^{\left(a\right)}=\frac{\left(1-\mu_t\right)^2\left(1-\mu_r\right)}{r\mu_r}\left[1-\left(1-\mu_r\right)^{K_r}-\frac{K_r}{2}\mu_r\left(1-\mu_r\right)^{K_r-1}\right]\\
\delta_F^{\left(b\right)}=\frac{\left(1-\mu_t\right)^2\left(1-\mu_r\right)}{r\mu_r}\Big[1-\left(1-\mu_r\right)^{K_r}\\
\qquad\qquad\qquad-\frac{K_r}{2}\mu_r\left(1-\mu_r\right)^{K_r-1}\left(\frac{1-3\mu_t}{1-\mu_t}\right)\Big]\\
\delta_E^{\left(a\right)}=\frac{\left(1-\mu_r\right)}{\mu_r}\Big[1-\left(1-\mu_r\right)^{K_r}\\
\qquad\qquad\qquad-\left(\frac{K_r}{2}-\mu_t\left(1-\mu_t\right)\right)\mu_r\left(1-\mu_r\right)^{K_r-1}\Big]\\
\delta_E^{\left(b\right)}=\frac{\left(1-\mu_r\right)}{\mu_r}\left[1-\left(1-\mu_r\right)^{K_r}-\frac{K_r}{2}\mu_r\left(1-\mu_r\right)^{K_r-1}\right]
\end{array}
\end{equation}
\normalsize
\end{theorem}
\begin{proof}
See Appendix~\ref{App1}
\end{proof}

%
%\begin{figure} [t]
%\centering
%\includegraphics[width=5 cm, height=4 cm]{GAP_K_10_r_1.eps}
%\caption{The gap between the decentralized NDT in~\ref{Th1} and the lower bound in~\ref{Th2} for the fronthaul capacity $r=1$ and the number of users $K_r=10$}
%\label{fig3}
%\end{figure}
%
%\begin{figure} [t]
%\centering
%\includegraphics[width=5 cm, height=4 cm]{GAP_K_100_r_1.eps}
%\caption{The gap between the decentralized NDT in~\ref{Th1} and the lower bound in~\ref{Th2} for the fronthaul capacity $r=1$ and the number of users $K_r=100$}
%\label{fig3}
%\end{figure}

\textbf{Lower Bound:} To evaluate the performance of the decentralized scheme, we derive a generic lower bound on the optimum NDT for the $K_t\times K_r$ F-RAN caching problem for serial and pipelined transmissions. 
\begin{theorem}~\label{Th2}  For an F-RAN with fronthaul capacity $C_F=r\log\left(P\right)$, $K_t$ ENs, each with a normalized cache size $\mu_t$, and $K_r$ users, each with a normalized cache size $\mu_r$, the optimum NDT $\delta^{*}_S$ for serial transmission is lower bounded by 
\begin{equation}
\delta^{*}_S\geq \delta_{LB}
\end{equation}
where $\delta_{LB}=\delta_F+\delta_E$ is the solution of the following linear programming
\begin{align}
&\min\ \delta_F+\delta_E\nonumber\\
\text{subject to} \ &s\delta_E+\left(K_t-s\right)r\delta_F\geq  f\left(s\right)\ s\in\mathcal{M}\label{cons1}\\
&\delta_F\geq 0,\ \delta_E\geq \left(1-\mu_r\right)
\end{align}
For pipelined transmission, the optimum NDT $\delta_P^{*}$ is lower bounded by

\small
\begin{equation}~\label{Lowerp}
\delta_P^{*}\geq\max\left\{\max_{s\in\left[0:\min\lbrace K_t,K_r\rbrace\right]}\left\{\frac{f\left(s\right)}{s+\left(K_t-s\right)r}\right\},\left(1-\mu_r\right)\right\}
\end{equation}\normalsize
In~\eqref{cons1} and~\eqref{Lowerp}, $\mathcal{M}\triangleq\left[0:\min\lbrace K_t,K_r\rbrace\right]$, and $f\left(s\right)= K_r\left(1-s\mu_r\right)-\left(K_r-s\right)\left[\left(K_r-s\right)\mu_r+\left(K_t-s\right)\mu_t\right]$.
\end{theorem} 
\begin{proof}
The proof of Theorem~\ref{Th2} uses the cut-set lower bound on any achievable NDT. The proof can be found in Appendix~\ref{App2}
\end{proof}

\vspace{-4mm}

\section{PERFORMANCE ANALYSIS}
In this section, we evaluate the performance of the decentralized scheme introduced in Section~\ref{Decentralized}. We first evaluate numerically the performance of the decentralized scheme for the general case. Then, we discuss two special cases.

Fig.~\ref{Fig2} shows the gap between the achievable decentralized NDT in Theorem~\ref{Th1} and the lower bound in Theorem~\ref{Th2} for serial and pipelined transmissions with fronthaul gain $r=1$ and a number of users $K_r=100$. It is observed that the maximum gap is about $20$ for serial transmission and $10$ for pipelined transmission. Moreover, we observe that the gap decreases rapidly as $\mu_r\rightarrow 1$, however, there is only a subtle decrease as $\mu_t$ increases.

\textbf{For $\mu_t=1$ or/and $r\rightarrow\infty$:} In this case, each EN has all the library files, and hence, the system becomes a MISO broadcast channel with two transmit antennas. From Theorem~\ref{Th1}, the decentralized scheme achieves NDT as

\small
\begin{equation}\label{eqn2}
\delta^{dec}\left(\mu_r\right)=\frac{\left(1-\mu_r\right)}{\mu_r}\left[1-\left(1-\mu_r\right)^{K_r}-\frac{K_r}{2}\mu_r\left(1-\mu_r\right)^{K_r-1}\right]
\end{equation}
\normalsize
Comparing with the performance of the broadcast channel with a single transmit antenna in~\cite{maddah2015decentralized}, NDT in~\eqref{eqn2} is minimum with term $\frac{K_r}{2}\left(1-\mu_r\right)^{K_r}$, since $K_r$ fragments cached at none of the users are transmitted with $DoF=2$ in our case as compared to $DoF=1$ in~\cite{maddah2015decentralized}. In Fig~\ref{Fig3}, we see that this reduction is significant when $\mu_r$ is low; however, the two systems achieve the same performance when $\mu_r$ approaches to one.

\textbf{For $\mu_r=0$:} The network becomes an F-RAN with caches at ENs only. In~\cite{sengupta2016cloud}, the authors provide a centralized coded caching scheme for this network within a constant factor $2$ from the lower bound. From Theorem~\ref{Th1}, the decentralized scheme for $\mu_r\rightarrow 0$ achieves 
\begin{IEEEeqnarray}{l}
\delta_S^{dec}=\left\{\begin{array}{ll}
\frac{K_r}{2}\left[\frac{\left(1-\mu_t\right)^2}{r}+1\right]+\mu_t\left(1-\mu_t\right) & 0<r\leq K_r\\
\frac{Kr}{2}\left[\frac{\left(1-\mu_t^2\right)}{r}+1\right] & r>K_r
\end{array}\right.~\label{eqn5}\\
\delta_P^{dec}=\left\{\begin{array}{ll}
\frac{K_r}{2} & r\geq r_1\\
\frac{Kr}{2r}\left(1-\mu_t^2\right) & r_2\leq r<r_1\\
\frac{Kr}{2r}\left(1-\mu_t\right)^2 & 0< r<r_2
\end{array}\right.\label{eqn4}
\end{IEEEeqnarray}
where \small$r_2=\left(1-\mu_t\right)^2/$ $\left(1+2\mu_t\left(1-\mu_t\right)/K_r\right)$, $r_1=\left(1-\mu_t^2\right)$.\normalsize
\begin{theorem}~\label{Th3}
For pipelined transmission, the decentralized scheme~\eqref{eqn4} is optimum for $\mu_t\in\left[0,1\right]$, $r\geq \left(1-\mu_t^2\right)$ and $\mu_t\in\left[0,1/2\right)$, $0<r$. Let $\delta^{*}_S$ be the optimal centralized caching scheme for serial transmission. We have $\frac{\delta_S^{dec}}{\delta_S^{*}}\leq 3$ for $\mu_t\in\left[0,1\right]$, $r\geq 1$ and $\mu_t\in\left[0,\sqrt{2}-1\right]$, $0<r$.
\end{theorem}
\begin{proof}
The proof is presented in Appendix~\ref{App3}
\end{proof}
Theorem~\ref{Th3} shows that the proposed decentralized scheme is approximately optimal for pipelined and serial transmissions. However, the optimality is not maintained in case \small$\mu_t\in\left[\sqrt{2}-1,1\right]$, $0<r<1$ \normalsize for serial transmission and \small$\mu_t\in\left[1/2,1\right]$, $0<r<\left(1-\mu_t^2\right)$ \normalsize for pipelined transmission. The reason is that when the fronthaul gain is low and the aggregate size of ENs' caches is high enough to store all the library files, the best strategy is to store different contents in ENs so that the fronthaul link would not be used. However, this strategy is not valid in the decentralized scheme due to the random content placement. \vspace{-1mm}
\section{CONCLUSION} 
In this paper, we have studied the problem of decentralized caching in Fog radio Access Networks. For two edge-nodes and arbitrary number of users, a coded delivery scheme has been developed to minimize the latency for delivering user demands. We have shown that even with the decentralized placement, the proposed delivery scheme achieves a significant performance as compared with the derived lower bound. 

\appendices
\section{PROOF OF THEOREM~\ref{Th1}}\label{App1}
%To prove the correctness of the delivery scheme, let us consider user $k\in\left[K_r\right]$. In order to recover the file $W_{d_k}$, the cloud and ENs should deliver all the fragments of file $W_{d_k}$ that are not available at the cache of the $k$th user, i.e., $W{d_k,\mathcal{S}_t,\mathcal{S}_r}$ for all $\mathcal{S}_t\subset \left[K_t\right]$ and $\mathcal{S}_r\subset\left[K_r\right]\setminus k$. For the proposed scheme, the $k$th user receives fragment $W_{d_k,\phi,\phi}$ in the first stage, and fragments $W_{d_k,\phi,\mathcal{S}_r\setminus k}$, $\mathcal{S}_r\subset$, in the second stage. Frag  

In Subsection~\ref{Coded}, we proposed two delivery schemes that have the same first fourth stages of transmission and differ only at the last stage. In order to obtain the fronthaul-NDT and edge-NDT presented in~\eqref{eqn3} for each scheme, we sum the over all fronthaul-NDT and edge-NDT of the five stages: 
\begin{equation}
\begin{aligned}
\delta_F^{\left(a\right)}&=\sum_{i=1}^{4}\delta_F^{\left(i\right)}+\delta_{F}^{\left(a,5\right)}\\
\delta_E^{\left(a\right)}&=\sum_{i=1}^{4}\delta_E^{\left(i\right)}+\delta_{E}^{\left(a,5\right)}
\end{aligned}
\end{equation}
The same analysis holds for the second delivery scheme $\left(b\right)$. For serial transmission, the NDT of the system is the sum of the fronthaul-NDT and edge-NDT (see~\eqref{ser}), so that ENs choose the delivery scheme that minimizes the summation of the fronthaul-NDT and edge-NDT. It is obvious that the first scheme $\left(a\right)$ has lower NDT than the second scheme $\left(b\right)$ as $r<K_r$. Thus, the NDT for serial transmission is given by~\eqref{eqn6}.

For pipelined transmission, the NDT of the system is the maximum between fronthaul-NDT and edge-NDT (see~\eqref{pipe}). Therefore, we first obtain the performance of each delivery scheme by taking the maximum between fronthaul-NDT and edge-NDT, then, we choose the delivery scheme that has the minimum NDT between them as in~\eqref{eqn7}.
 
\section{PROOF OF THEOREM~\ref{Th2}}\label{App2}
The lower bound is similar to the one presented in Proposition $1$ in~\cite{tandon2016cloud} which is based on the cut-set argument. Consider the user demands $\mathbf{d}$, where each user $k\in\left[K_r\right]$ requests a different file $W_{d_k}$. To obtain the constraint~\ref{cons1}, let us consider $s\in\left\{0,\cdots,\min\lbrace K_t,K_r\rbrace\right\}$ users. From the the received signals of $s$ users which are functions of the transmitted messages $\mathbf{X}_{\left[1:K_t\right]}$, the cache contents of $\left(K_t-s\right)$ ENs, and the fronthaul messages received by the $\left(K_t-s\right)$ ENs, we can reconstruct the transmitted messages $\mathbf{X}_{\left[1:K_t\right]}$ at high SNR regime by solving $K_t$ equations in $K_t$ unknowns. Therefore, we can argue that the $K_r$ user demands can be decodable by observing the received signals of $s$ users $\mathbf{Y}_{\left[1:s\right]}$, the cache contents of $\left(K_t-s\right)$ ENs $V_{\left[1:K_r-s\right]}$, the fronthaul messages of $\left(K_t-s\right)$ ENs $\mathbf{U}_{\left[1:K_t-s\right]}$, and the cache contents of $K_r$ users $Z_{\left[1:K_r\right]}$. \vspace{-2mm}

\small
\begin{IEEEeqnarray}{l}~\label{Lower1}
K_rF=H\left(W_{\left[1:K_r\right]}\right)\stackrel{\left(a\right)}{=}H\left(W_{\left[1:K_r\right]}|W_{\left[K_r+1:N\right]}\right)\nonumber\\
\stackrel{\left(b\right)}{=}I\left(W_{\left[1:K_r\right]};\mathbf{Y}_{\left[1:s\right]},Z_{\left[1:K_r\right]},V_{\left[1:K_t-s\right]},\mathbf{U}_{\left[1:K_t-s\right]}|W_{\left[K_r+1:N\right]}\right)\nonumber\\
+H\left(W_{\left[1:K_r\right]}|W_{\left[K_r+1:N\right]},\mathbf{Y}_{\left[1:s\right]},Z_{\left[1:K_r\right]},V_{\left[1:K_t-s\right]},\mathbf{U}_{\left[1:K_t-s\right]}\right)\nonumber\\
\stackrel{\left(c\right)}{\leq} I\left(W_{\left[1:K_r\right]};\mathbf{Y}_{\left[1:s\right]}|W_{\left[K_r+1:N\right]}\right)+\nonumber\\
I\left(W_{\left[1:K_r\right]};Z_{\left[1:K_r\right]},V_{\left[1:K_t-s\right]},\mathbf{U}_{\left[1:K_t-s\right]}|W_{\left[K_r+1:N\right]},\mathbf{Y}_{\left[1:s\right]}\right)+\varepsilon_F F\nonumber\\
\stackrel{\left(d\right)}{\leq} I\left(\mathbf{X}_{\left[1:K_t\right]};\mathbf{Y}_{\left[1:s\right]}\right)+I\left(W_{\left[1:K_r\right]};Z_{\left[1:K_r\right]}|W_{\left[K_r+1:N\right]},\mathbf{Y}_{\left[1:s\right]}\right)+\nonumber\\
I\left(W_{\left[1:K_r\right]};V_{\left[1:K_t-s\right]},\mathbf{U}_{\left[1:K_t-s\right]}|W_{\left[K_r+1:N\right]},\mathbf{Y}_{\left[1:s\right]},Z_{\left[1:K_r\right]}\right)+\varepsilon_F F\nonumber\\
\stackrel{\left(e\right)}{\leq} sT_E\log\left(P\right)+H\left(Z_{\left[1:K_r\right]}|W_{\left[K_r+1:N\right]},\mathbf{Y}_{\left[1:s\right]}\right)+\nonumber\\
I\left(W_{\left[1:K_r\right]};V_{\left[1:K_t-s\right]},\mathbf{U}_{\left[1:K_t-s\right]}|W_{\left[K_r+1:N\right]},\mathbf{Y}_{\left[1:s\right]},Z_{\left[1:K_r\right]}\right)+\varepsilon_F F\nonumber\\
=sT_E\log\left(P\right)+H\left(Z_{\left[1:K_r\right]}|W_{\left[K_r+1:N\right]},\mathbf{Y}_{\left[1:s\right]}\right)\nonumber\\
\quad+I\left(W_{\left[1:K_r\right]};V_{\left[1:K_t-s\right]}|W_{\left[K_r+1:N\right]},\mathbf{Y}_{\left[1:s\right]},Z_{\left[1:K_r\right]}\right)+\nonumber\\
 I\left(W_{\left[1:K_r\right]};\mathbf{U}_{\left[1:K_t-s\right]}|W_{\left[K_r+1:N\right]},\mathbf{Y}_{\left[1:s\right]},Z_{\left[1:K_r\right]},V_{\left[1:K_t-s\right]}\right)+\varepsilon_F F\nonumber\\
=sT_E\log\left(P\right)+H\left(Z_{\left[1:K_r\right]}|W_{\left[K_r+1:N\right]},\mathbf{Y}_{\left[1:s\right]}\right)\nonumber\\
\quad+H\left(V_{\left[1:K_t-s\right]}|W_{\left[K_r+1:N\right]},\mathbf{Y}_{\left[1:s\right]},Z_{\left[1:K_r\right]}\right)\nonumber\\
\quad +H\left(\mathbf{U}_{\left[1:K_t-s\right]}|W_{\left[K_r+1:N\right]},\mathbf{Y}_{\left[1:s\right]},Z_{\left[1:K_r\right]},V_{\left[1:K_t-s\right]}\right)+\varepsilon_F F\nonumber\\
\stackrel{\left(f\right)}{\leq}sT_E\log\left(P\right)+H\left(Z_{\left[1:K_r\right]}|W_{\left[K_r+1:N\right]},\mathbf{Y}_{\left[1:s\right]}\right)+\nonumber\\
\qquad \left(K_t-s\right)\left(K_r-s\right)\mu_tF +H\left(\mathbf{U}_{\left[1:K_t-s\right]}\right)+\varepsilon_F F\nonumber\\
\stackrel{\left(g\right)}{\leq}sT_E\log\left(P\right)+H\left(Z_{\left[1:K_r\right]}|W_{\left[K_r+1:N\right]},\mathbf{Y}_{\left[1:s\right]}\right)+\nonumber\\
\qquad \left(K_t-s\right)\left(K_r-s\right)\mu_tF +\left(K_t-s\right)T_Fr\log\left(P\right)+\varepsilon_F F
\end{IEEEeqnarray}

\normalsize

where $\left(a\right)$ follows from the fact that the files are independent. Step $\left(b\right)$ follows from the chain rule. Step $\left(c\right)$ follows from Fano's inequality, where \vspace{-2mm}

\small
\begin{equation*}
H\left(W_{\left[1:K_r\right]}|W_{\left[K_r+1:N\right]},\mathbf{Y}_{\left[1:s\right]},\mathbf{U}_{\left[1:K_t-s\right]},V_{\left[1:K_t-s\right]},Z_{\left[1:K_r\right]}\right)\leq \varepsilon_F F
\end{equation*}
\normalsize 
since the users are able to decode their requested files from their cache contents and the transmitted messages of ENs. $\varepsilon_F$ is a function of $F$ that vanishes as $F\rightarrow \infty$. Step $\left(d\right)$ follows from the data processing inequality, where $\mathbf{X}_{\left[1:K_t\right]}$ are functions of $W_{\left[1:K_r\right]}$, and the fact that the transmitted and received signals are independent on the unrequested files. Step $\left(e\right)$ follows from the capacity bound of $K_t\times s$ broadcast channel over $T_E$ time slots in the high SNR regime, where $s\leq K_t$, in addition to $H\left(Z_{\left[1:K_r\right],\mathbf{Y}_{\left[1:s\right]}}\mid W_{\left[1:N\right]}\right)=0$. Step $\left(f\right)$ follows from the fact that conditioning reduces entropy, and the fact that caches of $\left(K_t-s\right)$ ENs have $\left(K_t-s\right)\left(K_r-s\right)\mu_t F$ bits of $\left(K_r-s\right)$ files. Step $\left(g\right)$ follows from the capacity bound on $\left(K_t-s\right)$ fronthaul links over $T_F$ times slots. The second term in~\ref{Lower1} can be upper bounded as follows
\begin{equation}~\label{Lower2}
\begin{aligned}
&H\left(Z_{\left[1:K_r\right]}|W_{\left[K_r+1:N\right]},\mathbf{Y}_{\left[1:s\right]}\right)\stackrel{\left(a\right)}{=}H\left(Z_{\left[1:s\right]}|W_{\left[K_r+1:N\right]},\mathbf{Y}_{\left[1:s\right]}\right)\\
&+H\left(Z_{\left[s+1:K_r\right]}|W_{\left[K_r+1:N\right]},\mathbf{Y}_{\left[1:s\right]},Z_{\left[1:s\right]}\right)\\
&\stackrel{\left(b\right)}{\leq} H\left(Z_{\left[1:s\right]}|W_{\left[K_r+1:N\right]}\right)+H\left(Z_{\left[s+1:K_r\right]}|W_{\left[1:s\right]\cup\left[K_r+1:N\right]}\right)\\
&\stackrel{\left(c\right)}{=}sK_r\mu_rF+\left(K_r-s\right)^2\mu_rF
\end{aligned}
\end{equation}
where $\left(a\right)$ follows from the chain rule. Step $\left(b\right)$ follows from that conditioning reduces entropy, and $s$ users are able to decode files $W_{\left[1:s\right]}$ from their cache contents and the received signals. Step $\left(c\right)$ follows from that caches of $s$ users have $K_r\mu_rF$ bits of $K_r$ files, and caches of $\left(K_r-s\right)$ users have $\left(K_r-s\right)\mu_rF$ bits of $\left(K_r-s\right)$ files. Substituting from~\ref{Lower2} into~\ref{Lower1}, we get
\begin{equation}
\begin{aligned}
&\left(sT_E+r\left(K_t-s\right)T_F\right)\log\left(P\right)\geq K_r\left(1-s\mu_r\right)F\\
&-\left(K_r-s\right)\Big[\left(K_r-s\right)\mu_r+\left(K_t-s\right)\mu_t\Big]F+\varepsilon_F F
\end{aligned}
\end{equation}
By taking the limit $P\rightarrow\infty$ and the limit $F\rightarrow\infty$, we reach to the inequality~\ref{cons1} as
\begin{equation}
\begin{aligned}
s\delta_E+r&\left(K_t-s\right)\delta_F\geq K_r\left(1-s\mu_r\right)\\
&-\left(K_r-s\right)\left[\left(K_r-s\right)\mu_r+\left(K_t-s\right)\mu_t\right]
\end{aligned}
\end{equation}

To obtain the constraint $\delta_E\geq\left(1-\mu_r\right)$, consider the cut-set at the first user
\begin{equation}
\begin{aligned}
F&=H\left(W_1\right)=H\left(W_1|W_{\left[2:N\right]}\right)\\
&\stackrel{\left(a\right)}{=}I\left(W_1;Y_1,Z_1|W_{\left[2:N\right]}\right)+H\left(W_1|W_{\left[2:N\right]},Y_1,Z_1\right)\\
&\stackrel{\left(b\right)}{\leq}I\left(W_1;Y_1|W_{\left[2:N\right]}\right)+I\left(W_1;Z_1|W_{\left[2:N\right]},Y_1\right)+\varepsilon_F F\\
&\stackrel{\left(c\right)}{\leq} T_E \log\left(P\right)+H\left(Z_1|W_{\left[2:N\right]},Y_1\right)+\varepsilon_F F\\
&\stackrel{\left(d\right)}{\leq} T_E \log\left(P\right)+\mu_r+\varepsilon_F F\\
\end{aligned}
\end{equation}
where $\left(a\right)$ follows from the chain rule. Step $\left(b\right)$ follows from the Fano's inequality. Step $\left(c\right)$ follows from the capacity bound of the broadcast channel. Step $\left(d\right)$ follows from the fact that the cache of each user has $\mu_rF$ bits of each file. By taking limits $P\rightarrow\infty$, $L\rightarrow\infty$, we get $\delta_E\geq\left(1-\mu_r\right)$. This complete the proof of lower bound on the NDT for serial transmission.

Note that fronthaul messages $\mathbf{U}_{\left[K_t-s\right]}$, ENs' messages $\mathbf{X}_{\left[1:K_t\right]}$, and received messages $\mathbf{Y}_{\left[1:s\right]}$ cannot exceed $T_p$ time slots for pipelined transmission, where $T_p$ denotes end-to-end latency in case of pipelined transmission. Thus, the lower bound on the NDT for pipelined transmission can be obtained in a similar proof for serial transmission by substituting $T_E=T_F=T_p$ in~\eqref{Lower1}. By taking limits $P\rightarrow\infty$, $L\rightarrow\infty$ and combing with the fact that $\delta_P^{*}\geq\left(1-\mu_r\right)$, we obtain the lower bound presented in~\eqref{Lowerp}.

\section{PROOF OF THEOREM~\ref{Th3}}\label{App3}
For pipelined transmission, we have three inequalities on the minimum NDT $\left(\delta_P^{*}\right)$ from Theorem~\ref{Th2} by substituting $K_t=2$ and $\mu_r=0$ in~\eqref{Lowerp}, as follows
\begin{align}
&\delta_P^{*}\geq \frac{K_r\left(1-2\mu_t\right)}{2r}~\label{in1}\\
&\delta_P^{*}\geq \frac{K_r\left(1-\mu_t\right)+\mu_t}{1+r}~\label{in2}\\
&\delta_P^{*}\geq \frac{K_r}{2}~\label{in3}
\end{align}
We consider two main regimes of fronthaul gain $r$:

1) For $r\geq \left(1-\mu_t^2\right)$: It can be easily seen that the decentralized scheme in~\eqref{eqn4} achieves the lower bound from the inequality~\eqref{in3}.

2) For $0<r<\left(1-\mu_t^2\right)$ and $\mu_t\in\left[0,1/2\right)$: Using inequality~\eqref{in1}, we have
\begin{equation}
\frac{\delta_P^{dec}}{\delta_P^{*}}=\frac{\frac{K_r}{2r}\left(1-\mu_t\right)^2}{\frac{K_r}{2r}\left(1-2\mu_t\right)}=\frac{\left(1-\mu_t\right)^2}{\left(1-2\mu_t\right)}\leq 1
\end{equation} 
for $0<r<r_2$ and $\mu_t\in\left[0,1/2\right)$. Similarly, we have
\begin{equation}
\frac{\delta_P^{dec}}{\delta_P^{*}}=\frac{\frac{K_r}{2r}\left(1-\mu_t^2\right)}{\frac{K_r}{2r}\left(1-2\mu_t\right)}=\frac{\left(1-\mu_t^2\right)}{\left(1-2\mu_t\right)}\leq 1
\end{equation}  
for $r_2<r<\left(1-\mu_t^2\right)$ and $\mu_t\in\left[0,1/2\right)$. Combining with the fact that $\frac{\delta_P^{dec}}{\delta_P^{*}}\geq 1$, this proves the optimality of the decentralized scheme for pipelined transmission when $\mu_t\in\left[0,1\right]$, $r\geq \left(1-\mu_t^2\right)$ and $\mu_t\in\left[0,1/2\right)$, $0<r$.

From Theorem~\ref{Th2} by substituting with $\mu_r=0$ and $K_t=2$ in~\eqref{cons1}, we have $s\in\lbrace0,1,2\rbrace$ inequality constraints for serial transmission as follows
\begin{align}
&\delta_F\geq \frac{K_r}{2r}\left(1-2\mu_t\right)~\label{inq1}\\
&\delta_E+r\delta_F\geq K_r\left(1-\mu_t\right)+\mu_t~\label{inq2}\\
&\delta_E\geq \frac{K_r}{2}~\label{inq3}
\end{align}
We consider three different regimes of fronthaul gain $r$:

1) For $r\geq K_r$:  From the inequality~\eqref{inq3}, we obtain the lower bound on the minimum NDT for serial transmission
\begin{equation}~\label{reg1_2}
\delta_S^{*}\geq \frac{K_r}{2}
\end{equation}
From~\eqref{reg1_2} and~\eqref{eqn5}, we have
\begin{equation}
\frac{\delta_S^{dec}}{\delta_S^{*}}\leq \frac{\left(1-\mu_t^2\right)}{r}+1 \leq 2
\end{equation}
since $\left(1-\mu_t^2\right)\leq 1$ and $r\geq 1$, we have $\frac{\left(1-\mu_t^2\right)}{r}\leq 1$.

2) For $1\leq r < K_r$: From~\eqref{reg1_2} and~\eqref{eqn5}, we have
\begin{equation}
\frac{\delta_S^{dec}}{\delta_S^{*}}\leq \frac{\left(1-\mu_t\right)^2}{r}+1+\frac{2\mu_t\left(1-\mu_t\right)}{K_r} \leq 3
\end{equation}
where $\frac{\left(1-\mu_t\right)^2}{r}\leq 1$ as $\left(1-\mu_t\right)^2\leq 1$ and $r\geq 1$. In addition, $\frac{2\mu_t\left(1-\mu_t\right)}{K_r}\leq 1$, since $K_r\geq 2$ and $\mu_t\left(1-\mu_t\right)\leq 1$.

3) For $0<r<1$, $\mu_t\in\left[0:\sqrt{2}-1\right]$: From the inequalities~\ref{inq1},~\ref{inq2}, we have
\begin{equation}~\label{reg2}
\delta_S^{*}\geq \frac{K_r}{2r}\left(1-2\mu_t\right)+\frac{K_r}{2}+\mu_t
\end{equation}
Using~\eqref{eqn5} and the lower bound~\eqref{reg2}, we obtain
\begin{equation}
\begin{aligned}
\frac{\delta_S^{dec}}{\delta_S^{*}}&\leq \frac{\frac{K_r}{2}\left[\frac{\left(1-\mu_t\right)^2}{r}+1\right]+\mu_t\left(1-\mu_t\right)}{\frac{K_r}{2r}\left(1-2\mu_t\right)+\frac{K_r}{2}+\mu_t}\\
&=1+\frac{\mu_t^2\frac{K_r}{2r}-\mu_t^2}{\frac{K_r}{2r}\left(1-2\mu_t\right)+\frac{K_r}{2}+\mu_t}\leq 2
\end{aligned}
\end{equation}
where $\frac{\mu_t^2\frac{K_r}{2r}-\mu_t^2}{\frac{K_r}{2r}\left(1-2\mu_t\right)+\frac{K_r}{2}+\mu_t}\leq 1$ for $\mu_t\in \left[0:\sqrt{2}-1\right]$. This completes the proof of Theorem~\ref{Th3}.


\end{document}